\documentclass[twocolumn,journal]{IEEEtran}

\usepackage{amsfonts, amssymb, amsmath, cite, enumerate, amsthm, bm, caption, subfig, psfrag, cases, mathtools,cleveref}
\usepackage[usenames,dvipsnames]{color}
\ifCLASSINFOpdf
   \usepackage[pdftex]{graphicx}
   \DeclareGraphicsExtensions{.pdf,.jpeg,.png}
\else
   \usepackage[dvips]{graphicx}
   \DeclareGraphicsExtensions{.eps}
\fi

\DeclareMathOperator{\E}{\mathbb{E}}

\DeclareMathOperator*{\diag}{diag}

\DeclareMathOperator*{\tr}{tr}

\newcommand{\ip}[2]{\left\langle#1,#2\right\rangle}
\newcommand{\norm}[1]{\left\lVert#1\right\rVert}
\renewcommand{\(}{\left(}
\renewcommand{\)}{\right)}
\renewcommand{\[}{\left[}
\renewcommand{\]}{\right]}

\def \P {\mathbb{P}}
\def \R {\mathbb{R}}

\def \e {\varepsilon}

\def \l {\lambda}

\def \< {\langle}
\def \> {\rangle}

\def \va {\bm{a}}
\def \vb {\bm{b}}

\def \vx {\bm{x}}
\def \vu {\bm{u}}

\def \vy {\bm{y}}

\def \vX {\bm{X}}
\def \vY {\bm{Y}}
\def \vA {\bm{A}}
\def \vB {\bm{B}}
\def \vD {\bm{D}}

\def \vW {\bm{W}}

\def \vI {\bm{I}}
\def \vU {\bm{U}}

\def \vP {\bm{P}}

\def \vT {\bm{T}}

\def \vSigma {\bm{\Sigma}}

\def \vLambda {\bm{\Lambda}}
\def \vTheta  {\bm{\Theta}}

\theoremstyle{plain}
\newtheorem{theorem}{Theorem}
\newtheorem{proposition}{Proposition}
\newtheorem{corollary}{Corollary}
\newtheorem{definition}{Definition}

\newtheorem{fact}{Fact}

\theoremstyle{remark}
\newtheorem{remark}{Remark}
\newtheorem{example}{Example}

\begin{document}
\title{Covariance Matrix Estimation from Linearly-Correlated Gaussian Samples}

\author{Wei~Cui, Xu~Zhang, and Yulong~Liu
\thanks{This work was supported by the National Natural Science Foundation of China under Grants 61301188 and 61672097. (Corresponding author: Yulong Liu.)}
\thanks{W.~Cui and X.~Zhang are with the School of Information and Electronics, Beijing Institute of Technology, Beijing 100081, China (e-mail: cuiwei@bit.edu.cn; connorzx@bit.edu.cn).}
\thanks{Y.~Liu is with the School of Physics, Beijing Institute of Technology, Beijing 100081, China (e-mail: yulongliu@bit.edu.cn).}
}

\maketitle

\begin{abstract}
Covariance matrix estimation concerns the problem of estimating the covariance matrix from a collection of samples, which is of extreme importance in many applications. Classical results have shown that $O(n)$ samples are sufficient to accurately estimate the covariance matrix from $n$-dimensional independent Gaussian samples. However, in many practical applications, the received signal samples might be correlated, which makes the classical analysis inapplicable. In this paper, we develop a non-asymptotic analysis for the covariance matrix estimation from { linearly-}correlated Gaussian samples. Our theoretical results show that the error bounds are determined by the signal dimension $n$, the sample size $m$, and the shape parameter of the distribution of the correlated sample covariance matrix. Particularly, when the shape parameter is a class of Toeplitz matrices (which is of great practical interest), $O(n)$ samples are also sufficient to faithfully estimate the covariance matrix from correlated samples. Simulations are provided to verify the correctness of the theoretical results.
\end{abstract}


%
\IEEEpeerreviewmaketitle

\section{Introduction}

\IEEEPARstart{E}{stimating} covariance matrices becomes fundamental problems in modern multivariate analysis, which finds applications in many fields, ranging from signal processing\cite{krim1996two} and machine learning \cite{friedman2001elements} to statistics \cite{cai2016estimating} and finance \cite{fan2016overview}. In particular, important examples in signal processing include Capon's estimator \cite{capon1969high}, MUltiple SIgnal Classification (MUSIC) \cite{schmidt1986multiple}, Estimation of Signal Parameter via Rotation Invariance Techniques (ESPRIT) \cite{roy1989esprit}, and their variants \cite{krim1996two}.

During the past few decades, there have been numerous works devoted to studying the optimal sample size $m$ that suffices to estimate the covariance matrix from $n$-dimensional independent samples \cite{bai1993limit,aubrun2007sampling,vershynin2010into,adamczak2010quantitative,adamczak2011sharp,vershynin2012close,srivastava2013covariance,koltchinskii2017concentration}. For instance, Vershynin \cite{vershynin2010into} has shown that $m=O(n)$ samples are sufficient for independent sub-Gaussian samples, where $O(n)$ denotes that the order of growth of the samples is a linear function of the dimension $n$; Vershynin \cite{vershynin2012close} also illustrates that $O(n\log n)$ samples are required for independent heavy tailed samples; and Srivastava et al. \cite{srivastava2013covariance} have established that $O(n)$ is the optimal bound for independent samples which obey log-concave distributions.

However, in many practical applications of interest, it is very hard to ensure that the received signal samples are independent of each other. For example, in signal processing, the signal sources might be in multipath channel \cite{ramirez2010detection,huang2013detection} or interfere with each other \cite{shiu2000fading,liu2007training}, which causes the received samples correlated. In  portfolio management and risk assessment, the returns between different assets are correlated on short time scales, i.e., the Epps effect \cite{epps1979comovements,munnix2010impact}. Then a natural question to ask is:

\emph{Is it possible to use correlated samples to estimate the covariance matrix? If possible, how many correlated samples do we need to obtain a good estimation of the covariance matrix?}

This paper focuses on the above question and provides some related theoretical results. More precisely, we establish non-asymptotic error bounds for covariance matrix estimation from { linearly-}correlated Gaussian samples in both expectation and tail forms. These results show that the error bounds are determined by the signal dimension $n$, the sample size $m$, and the shape parameter $\vB$ of the distribution of the correlated sample covariance matrix. {  In particular, if the shape parameter is a class of Toeplitz matrices (see Section \ref{sec: CECC}.C Example 2 for details), where the shape parameter $\vB$ satisfies $\tr(\vB)=m$, $||\vB||_{F}=O(m^{1/2})$, and $||\vB||=O(1)$,} our results reveal that the correlated case has the same order of error rate as the independent case albeit with a larger multiplicative coefficient.

The remainder of this paper is organized as follows. The problem is formulated in Section \ref{sec: problem formulation}. The performance analysis of the covariance matrix estimation from { linearly-}correlated Gaussian samples is presented in Section \ref{sec: CECC}. Simulations are provided in Section \ref{sec: Simulation}, and conclusion is drawn in Section \ref{sec: Conclusion}.

\section{problem formulation} \label{sec: problem formulation}
 Let $\vx \in \R^n$ be a centered Gaussian vector with the covariance matrix $\vSigma=\E [\bm{x} \bm{x}^T]$, where $\vSigma \in \R^{n \times n}$ is a positive definite matrix.  Let $\vx_1,\ldots,\vx_m \in \R^n$ be independent copies of $\vx$. Suppose we observe $m$ {linearly}-correlated samples $\{\bm{y}_k\}_{k=1}^{m}$
 $$
 \vY=\vX \bm{\Lambda},
 $$
where $\bm{Y}=[\vy_1,\ldots,\vy_m]$, $\vX=[\vx_1,\ldots,\vx_m]$, and   $\vLambda \in \R^{m \times m}$  is a fixed matrix. The objective is to estimate the covariance matrix
$\vSigma$ from correlated samples $\{\bm{y}_k\}_{k=1}^{m}$. Here we assume that $m \ge n$
\footnote {{When $m < n$, in order to estimate covariance matrices, we require some kinds of prior information about covariance matrices. During the past few decades, a large number of estimators have been proposed to solve the problems in this setting. For example, one group is structured estimators, which impose additional structures on covariance matrices (see \cite{cai2016estimating} and references therein). Typical examples of structured covariance matrices include bandable covariance matrices \cite{bickel2008regularized}, Toeplitz covariance matrices \cite{cai2013optimal}, sparse covariance matrices \cite{bickel2008covariance,rothman2009generalized} and so on. Another group is to shrinkage the sample covariance matrix to a ``target'' matrix by incorporating some regularization \cite{carlson1988covariance,ledoit2004well,chen2010shrinkage,coluccia2015regularized}. The general form of shrinkage estimators is
\begin{equation*}
\hat{\vSigma}^{\text{SH}}=\alpha \vSigma_0 + (1-\alpha) \hat{\vSigma},
\end{equation*}
where $\vSigma_0 \in \R^{n \times n}$ is the shrinkage ``target'' matrix with positive definite structure, $\hat{\vSigma}$ denotes the sample covariance matrix, and $\alpha \in [0,1]$ is an absolute constant.
The third group of estimators is based on spectrum correction. In this group, spectrum correction approaches are utilized to infer a mapping from the sample eigenvalues to corrected eigenvalue estimates which yield a superior covariance matrix, see e.g., \cite{stein1986lectures,el2008spectrum,ledoit2011eigenvectors}.}
}.

The standard approach under correlated samples utilizes the correlated sample covariance matrix to approximate the actual one (see, e.g., \cite{collins2013compound, burda2011applying, chuah2002capacity})
\begin{equation} \label{eq: SCM}
	\hat{\vSigma}=\frac{1}{m} \sum \limits_{k=1}^m \vy_k \vy_k^T = \frac{1}{m}\vX \vLambda \vLambda^T \vX^T := \frac{1}{m}\vX \vB \vX^T.
\end{equation}
Our problem then becomes to investigate how many correlated samples are enough to estimate $\vSigma$ accurately from $\hat{\vSigma}$. It is not hard to find that the correlated sample covariance matrix $\hat{\vSigma}$ is a compound Wishart matrix\footnote{Let $\vx_1,\ldots,\vx_m \sim \mathcal{N}(\bm{0},\vSigma)$ be independent Gaussian vectors, and let $\vB$ be an arbitrary real $m \times m$ matrix. We say that a random $n \times n$ matrix $\vW$ is a compound Wishart matrix  with shape parameter $\vB$ and scale parameter $\vSigma$ if $\vW = \frac{1}{m}\vX \vB \vX^T$, where $\vX = [\vx_1,\ldots,\vx_m]$.} with shape parameter $\vB = \vLambda \vLambda^T$ and scale parameter $\vSigma$ \cite{speicher1998combinatorial}.

For the convenience of comparison, we restate a typical result for covariance matrix estimation from independent Gaussian samples as follows. This result indicates that $O(n)$ samples are sufficient to estimate the covariance matrix accurately from independent Gaussian samples. It is natural to expect that we require at least $O(n)$ samples to estimate the covariance matrix from correlated samples.

\begin{proposition}[Theorem 4, \cite{koltchinskii2017concentration}] \label{pp: sg}
 Let $\vx$ be a centered $n$-dimensional Gaussian vector with the covariance matrix $\vSigma=\E[\vx \vx^T]$,  and let $\vx_1,\ldots,\vx_m \in \R^n$ be independent copies of $\vx$. Then the sample covariance matrix $\widetilde{\vSigma}=\frac{1}{m} \sum_{k=1}^m \vx_k \vx_k^T$ satisfies
$$\E||\widetilde{\vSigma}-\vSigma|| \le C \( \sqrt{\frac{n}{m}}+\frac{n}{m}\) \norm{\vSigma},$$
where $\norm{\cdot}$ denotes the spectral norm and $C$ is an absolute constant.
\end{proposition}

\section{Covariance Matrix Estimation From { Linearly-}Correlated Gaussian Samples} \label{sec: CECC}

In this section, we present our main results for the covariance matrix estimation from { linearly-}correlated Gaussian samples. Our proof strategy is divided into two steps. First, we establish a key theorem which illustrates that the correlated sample covariance matrix
$$
\hat{\vSigma}=\frac{1}{m}\vY\vY^T=\frac{1}{m}\vX \vB \vX^T
$$
concentrates around its mean $\E \hat{\vSigma}$ with high probability. We then establish the non-asymptotic error bounds for the estimated covariance matrix in both expectation and tail forms.

\subsection{Concentration of {Linearly-}Correlated Sample Covariance Matrix}
\begin{theorem} \label{lm: wishart1}
	Let $\vx_1,\ldots,\vx_m \sim \mathcal{N}(\bm{0},\vSigma)$ be independent Gaussian vectors, where $\vSigma$ is an $n \times n$ real positive definite matrix. Let $\vB$ be a fixed symmetric real $m \times m$ matrix. Consider the compound Wishart matrix $\vW= \vX \vB \vX^T/m$ with $\vX=[\vx_1,\ldots,\vx_m]$. {Then for any $\delta \geq 0$, the following event
	$$\norm{\vW-\E\vW}  \le \frac{32\norm{\vB}_F\delta+64\norm{\vB}\delta^2}{m} \norm{\vSigma}$$
	holds with probability at least $1-2\exp(-2\delta^2+2n\log3)$, where $\norm{\cdot}_{F}$ denotes the Frobenius norm.}
	Furthermore,
	\begin{equation} \label{neq: E1}
	\E\norm{\vW-\E\vW } \le \frac{72 \norm{\vB}_{F} \sqrt{n} + 282 \norm{\vB} n}{m}\norm{\vSigma}.
	\end{equation}
\end{theorem}
\begin{proof}
  See { Appendix \ref{sec: Proof1}}.
\end{proof}

\begin{remark}
	It follows from Theorem \ref{lm: wishart1} that the error bounds depend on the signal dimension $n$, the sample size $m$, and the shape parameter $\vB$. In particular, if $||\vB||_{F}=O(m^{1/2})$ and $||\vB||=O(1)$, then this result reveals that $m = O(n)$ samples are sufficient to estimate the compound Wishart matrix $\vW$ accurately.
\end{remark}

\begin{remark}[Symmetric $\vB$]
	The fact that the shape matrix $\vB$ of the correlated sample covariance matrix is symmetric plays a key role in the proof of Theorem \ref{lm: wishart1}. This property enables the compound Wishart matrix $\vW$ to be expressed as a \emph{weighted} sum of independent rank-one matrices. Thus we can employ standard techniques in \cite{vershynin2010into} (e.g., $\varepsilon$-net method and Bernstein's inequality) to establish the error bounds in both expectation and tail forms.

\end{remark}

\begin{remark}[General $\vB$]
For general $\vB$, however, the compound Wishart matrix $\vW$ cannot be expressed as an independent weighted sum, which makes the theoretical analysis much harder.

In \cite{soloveychik2014error}, Soloveychik closely follows a sophisticated strategy developed by Levina and Vershynin \cite{levina2012partial} and establishes the following expectation bound
\begin{multline*}
  \E \norm{\vW-\E\vW}  \\ \le \frac{24 \lceil\log 2n\rceil^2 \sqrt{n} (4\|\vB\| + \sqrt{\pi} \|\vB\|_{F}/\|\vB\|)}{m} \norm{\vSigma}.
\end{multline*}
It is not hard to see that if $||\vB||_{F}=O(m^{1/2})$ and $||\vB||=O(1)$, then this bound shows that $m = O(n\log^4n)$ samples are { sufficient} to estimate the compound Wishart matrix $\vW$ accurately.

In \cite{paulin2016efron}, Paulin et al. employ the method of exchangeable pairs \cite{stein1972bound, stein1986approximate} and establish the concentration of $\vW$ in both expectation and tail forms for the bounded sample matrix $\vX$ (i.e., each entry of $\vX$ is bounded by an absolute positive constant $L$). The expectation bound in \cite{paulin2016efron} is given by
\begin{equation*}\label{leq: Paulin}
	\E ||\vW- \E\vW||  \le  \frac{ 2 \sqrt{ v(\vB) \log n } + 32\sqrt{3} L n \log n ||\bm{B}||}{m},
\end{equation*}
where $v(\vB)=44(n \sigma^2+L^2)\|\vB\|_{F}^2$ and $\sigma$ is the standard deviation of each entry of $\vX$. Clearly, if $||\vB||_{F}=O(m^{1/2})$ and $||\vB||=O(1)$, then this bound establishes that $m = O(n\log n)$ samples suffice to estimate the compound Wishart matrix $\vW$.

In contrast to the above two works, our proof strategy is totally different from theirs. This is because we have exploited the symmetric structure of $\vB$. More importantly, our results improve theirs in the symmetric case. This improvement is critical to obtain the optimal error rate for the covariance matrix estimation from correlated samples.

\end{remark}

\begin{remark}
It is worth pointing out that there is a different line of research which studies the asymptotic behavior of the compound Wishart matrix (e.g., $m \rightarrow \infty$ or $m, n \rightarrow \infty$). Please refer to \cite{bryc2008compound} and references therein for a survey.
\end{remark}

\subsection{Covariance Matrix Estimation from { Linearly-}Correlated { Gaussian} Samples}
We then derive the error bounds for the covariance matrix estimation from { linearly-}correlated { Gaussian} samples.

\begin{theorem} \label{thm: General} Let $\vx_1,\ldots,\vx_m \sim \mathcal{N}(\bm{0},\vSigma)$ be independent random vectors,  where $\vSigma$ is an $n \times n$ real positive definite matrix. Let $\bm{X}=[\vx_1,\ldots,\vx_m] \in \R^{n \times m}$. Consider the correlated samples $\vY=[\vy_1,\ldots,\vy_m]=\vX \vLambda $, where $\vLambda \in \R^{m \times m}$ is a fixed matrix. Let the sample covariance matrix $\hat{\vSigma} = \frac{1}{m} \sum_{k=1}^m \vy_k \vy_k^T$. {Then for any $\delta \geq 0$, the event
	\begin{multline*}
	\norm{\hat{\vSigma}-\vSigma} \le \left| \frac{\tr(\vLambda \vLambda^T)}{m}-1 \right|  ||\vSigma||
	\\ + \frac{32\norm{\vLambda \vLambda^T}_F\delta+64\norm{\vLambda \vLambda^T} \delta^2}{m}\norm{\vSigma}
	\end{multline*}
	holds with probability at least $1-2\exp(-2\delta^2+2n\log3)$.} Furthermore,
	\begin{multline*}
	\E \norm{\hat{\vSigma}-\vSigma} \le  \left| \frac{\tr(\vLambda \vLambda^T)}{m}-1 \right|  ||\vSigma|| \\
	+\frac{72 \norm{\vLambda \vLambda^T}_{F} \sqrt{n} + 282 \norm{\vLambda \vLambda^T} n}{m} \norm{\vSigma}.
	\end{multline*}
\end{theorem}

\begin{IEEEproof} By the triangle inequality, we have
	\begin{align} \label{mthm: neq_1}
	\mathbb{E}\norm{\hat{\vSigma}-\vSigma}  &=\E \norm{\frac{1}{m}\vY \vY^T-{\vSigma}}\nonumber\\
	&\le \E \norm{\frac{1}{m}\vY \vY^T-\frac{1}{m}\E[\vY \vY^T]} \nonumber\\
	& \quad + \norm{\frac{1}{m}\E[\vY \vY^T]-{\vSigma}}.
	\end{align}
	The first term in \eqref{mthm: neq_1} can be easily bounded by Theorem \ref{lm: wishart1}, i.e.,
	\begin{multline} \label{mthm: neq_2}
	\E \norm{\frac{1}{m}\vY \vY^T-\frac{1}{m}\E[\vY \vY^T]}\\ \le \frac{72 \norm{\vLambda \vLambda^T}_{F} \sqrt{n} + 282 \norm{\vLambda \vLambda^T} n}{m} \norm{\vSigma}.
	\end{multline}
    It suffices to bound the second term in \eqref{mthm: neq_1}. Since the columns of $\vX$ are centered independent Gaussian vectors, direct calculation leads to
	\begin{align*}
	\E[\(\vY \vY^T\)_{ij}]&=\sum_{l,k=1}^{m} (\vLambda\vLambda^T)_{lk} \E \(X_{il} X_{jk}\)\\
	&=\sum_{l=1}^{m} (\vLambda\vLambda^T)_{ll} \E \(X_{il} X_{jl}\)\\
	&=\tr(\vLambda\vLambda^T) \vSigma_{ij},
	\end{align*}
	where $X_{ij}$ denotes the $(i,j)$-th entry of the matrix $\vX$, $i=1,\ldots,n, j=1,\ldots,m$. Thus we have
	\begin{equation} \label{mthm: neq_3}
	\E\[\frac{1}{m}\vY \vY^T\]=\frac{\tr(\vLambda \vLambda^T)}{m} \vSigma.
	\end{equation}
	Substituting \eqref{mthm: neq_2} and \eqref{mthm: neq_3}  into \eqref{mthm: neq_1} yields the expectation bound.

    To establish the tail bound, observe that
    \begin{align*}
		\P \(\norm{\hat{\vSigma}-\vSigma} \ge t\) &\le \P \(\norm{\hat{\vSigma}-\E\hat{\vSigma}} + \norm{\E\hat{\vSigma}-{\vSigma}}\ge t\) \\
		& = \P \(\norm{\hat{\vSigma}-\E\hat{\vSigma}} \ge t-\norm{\E\hat{\vSigma}-{\vSigma}}\).
	\end{align*}
	Assign
	\begin{align*}
	t_0 &=\norm{\E\hat{\vSigma}-{\vSigma}} + \frac{32\norm{\vLambda \vLambda^T}_F\delta+64\norm{\vLambda \vLambda^T} \delta^2}{m}\norm{\vSigma}\\
	&= \left| \frac{\tr(\vLambda \vLambda^T)}{m}-1 \right|  ||\vSigma||\\
	&\qquad \qquad \qquad  +
	\frac{32\norm{\vLambda \vLambda^T}_F\delta+64\norm{\vLambda \vLambda^T} \delta^2}{m}\norm{\vSigma}.
	\end{align*}
	It then follows from Theorem \ref{lm: wishart1} that for any $\delta \geq 0$
	$$
	\P \(\norm{\hat{\vSigma}-\vSigma} \ge t_0\) \le 2 \exp(-2\delta^2 + 2n\log3).
	$$	
	This completes the proof.

\end{IEEEproof}

In particular, if the shape matrix satisfy $\tr(\vLambda \vLambda^T)=m$ (see examples in Section \ref{examples}), then we have following corollary.

\begin{corollary}   \label{col: Main Theorem}

Let $\vx_1,\ldots,\vx_m \sim \mathcal{N}(\bm{0},\vSigma)$ be independent random vectors,  where $\vSigma$ is an $n \times n$ real positive definite matrix. Let $\bm{X}=[\vx_1,\ldots,\vx_m] \in \R^{n \times m}$. Consider the correlated samples $\vY=[\vy_1,\ldots,\vy_m]=\vX \vLambda $, where $\vLambda \in \R^{m \times m}$ is a fixed matrix such that $\tr(\vLambda \vLambda^T)=m$. Let the sample covariance matrix $\hat{\vSigma} = \frac{1}{m} \sum_{k=1}^m \vy_k \vy_k^T$. {Then for any $\delta \geq 0$, the event
$$\norm{\hat{\vSigma}-\vSigma}  \le \frac{32\norm{\vLambda \vLambda^T}_F\delta+64\norm{\vLambda \vLambda^T} \delta^2}{m}\norm{\vSigma}$$
holds with probability at least $1-2\exp(-2\delta^2+2n\log3)$.} Furthermore,
\begin{equation*}
\E \norm{\hat{\vSigma}-\vSigma} \le  \frac{72 \norm{\vLambda \vLambda^T}_{F} \sqrt{n} + 282 \norm{\vLambda \vLambda^T} n}{m} \norm{\vSigma}.
\end{equation*}

\end{corollary}

\subsection{Examples}\label{examples}
In this subsection, we present some examples to illustrate our theoretical results.

\begin{example}[Independent samples]
	In this case, $\vLambda \vLambda^T=\vI_m$, where $\vI_m$ is the $m$-dimensional identity matrix. It is easy to verify that $\tr(\vI_m)=m$, $\norm{\vI_m}_{F}=\sqrt{m}$, and  $\norm{\vI_m}=1$. It then follows from Corollary \ref{col: Main Theorem} that
	\begin{equation} \label{leq: I}
	\mathbb{E}\norm{\hat{\vSigma}-\vSigma} \le \( 72 \sqrt{\frac{n}{m}} + \frac{282 n}{m} \) \norm{\vSigma}.
	\end{equation}
	It is clear that $m=O(n)$ samples are sufficient to estimate the covariance matrix in this case. This result is consistent with Proposition \ref{pp: sg}.
\end{example}

\begin{example}[Partially correlated samples] In this case, a typical model for the shape parameter is that $\vLambda \vLambda^T$ is a symmetric Toeplitz matrix
	\begin{equation*}
	\vLambda \vLambda^T=\left[
	\begin{array}{cccccc}
	1      & \theta      & \cdots & \theta^{m-1}  \\
	\theta      & 1      & \ddots & \vdots   \\
	\vdots & \ddots & \ddots & \theta        \\
	\theta^{m-1}& \cdots & \theta      & 1        \\
	\end{array}
	\right]:=\vT(\theta)
	\end{equation*}
	 with $0 < \theta < 1$. This model is very common in many applications. For instance, the lagged correlation between the returns in portfolio optimization \cite{burda2011applying} satisfies this model by setting $\theta=\exp(-1/\tau)$, where $\tau$ is the characteristic time. Obviously, $\tr(\vT(\theta))=m$. By Gershgorin circle theorem \cite[Theorem 7.2.1]{golub2012matrix}, we have
	\begin{equation*}
	\norm{\vT(\theta)} \le 1+2 \cdot \sum_{k=1}^{\infty} \theta^k=1+\frac{2\theta}{1-\theta}=\frac{1+\theta}{1-\theta},~ 0 < \theta < 1.
	\end{equation*}
	In addition,
	\begin{align*}\label{eq: T_F}
	||\vT(\theta)||_F^2 &=m+ 2 \cdot \sum_{k=1}^{m-1}(m-k)\theta^{2k}\\
	&=\frac{m(1+\theta^2)}{1-\theta^2} + \frac{2\theta^2(\theta^{2m}-1)}{(1-\theta^2)^2}  \le \frac{m(1+\theta^2)}{1-\theta^2},
	\end{align*}
    where that last inequality holds because $0 < \theta < 1$.
    It also follows from Corollary \ref{col: Main Theorem} that
	\begin{equation} \label{leq: T}
	\mathbb{E}\norm{\hat{\vSigma}-\vSigma} \le \(72 \sqrt{\frac{1+\theta^2}{1-\theta^2} \cdot \frac{n}{m}} + 282 \cdot \frac{1+\theta}{1-\theta} \cdot\frac{n}{m} \) \norm{\vSigma}.
	\end{equation}
	Therefore, we conclude that in this case, $m=O(n)$ correlated samples are also sufficient to accurately estimate the covariance matrix. The difference between the correlated case and the independent case is that for a given estimation accuracy of $\vSigma$, the former requires more samples than that of the latter. This is because the multiplicative coefficient in the error bound \eqref{leq: T} is larger than that in \eqref{leq: I}. Furthermore, the larger the $\theta$, the greater the multiplicative coefficient.
	
\end{example}

\begin{example}[Totally correlated samples]
	When the observed signal samples are totally correlated, for example, $\vy_k=\frac{1}{\sqrt{m}}\sum_{i=1}^{m}\vx_i$ for $k=1,\ldots,m$, which means that $\vLambda \vLambda^T$ is an all-one matrix
	\begin{equation*}
	\vLambda \vLambda^T=\left[
	\begin{array}{cccccc}
	1      & \cdots & 1        \\
	\vdots & \ddots & \vdots         \\
	1      & \cdots  & 1        \\
	\end{array}
	\right]:=\vTheta.
	\end{equation*}
	Standard calculation shows that $\tr(\vTheta)=m$, $\norm{\vTheta}_F=m$, and $\norm{\vTheta}=m$. By Corollary \ref{col: Main Theorem}, we have
	\begin{equation*}
	\E \norm{\hat{\vSigma}-\vSigma} \le  \(72 \sqrt{n} + 282 n \)\norm{\vSigma}.
	\end{equation*}
	This result indicates that when the samples are totally correlated, the error bounds are independent of the sample size $m$, which means that increasing $m$ will not reduce the estimation error.

\end{example}

\section{Simulation Results} \label{sec: Simulation}
In this section, we carry out some simulations to demonstrate our theoretical results.

Consider an $n \times m$ matrix $\vX$ whose entries are independently drawn from the standard Gaussian distribution. Let $\vY \in \R^{n \times m} $ satisfy $\vY=\vX \vLambda$.

{In the first experiment, we consider the case where the samples are independent but with time-variant scale factors, i.e., $\vLambda$ is a diagonal matrix with different diagonal entries. Let $\vLambda = \vP(\mu,\sigma) = \diag\{ \rho_1,\ldots,\rho_m\}$, where $\{\rho_i\}_{i=1}^m$ have independent Gaussian distribution with mean $\mu$ and standard deviation $\sigma$. We make simulations for four models: 1) $\vLambda =\vI$; 2) $\vLambda =\vP(\sqrt{3}/2,1/2)$; 3) $\vLambda =\vP(\sqrt{2}/2,\sqrt{2}/2)$; 4) $\vLambda =\vP(0,1)$. It is not hard to verify the four models satisfy  $\E\vLambda \vLambda^T=\vI$ and $\E\tr(\vLambda \vLambda^T)=m$.  We set $\eta=0.2$ and increase $n$ from 1 to 30. For a fixed $n$, we make $500$ trials and calculate the average of the minimum sample size $m$ which satisfies
$$
\frac{||\hat\vSigma - \vSigma||_F}{ ||\vSigma||_F} \le \eta.
$$

Fig. \ref{fig: sim1} shows the simulation results. It is not hard to find that the sample size is proportional to the signal dimension $n$ for the four models. With the increase of the standard deviation, the slope of the line also increases. This phenomenon can be explained by Theorem \ref{thm: General}: when the standard deviation increases, the average of both $\norm{\vLambda \vLambda^T}_F$ and $\norm{\vLambda \vLambda^T}$ will also increase, which leads to the increase of the slope.
\begin{figure}
	\centering
	\includegraphics[width=3.2in]{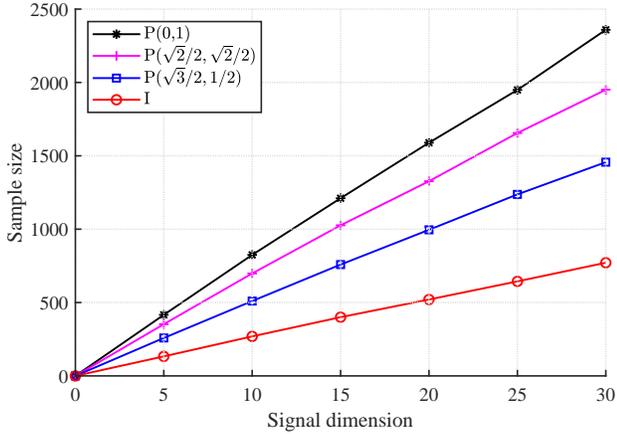}
	\caption{Sample size v.s. signal dimension for different independent models with time-variant scale factors.}
	\label{fig: sim1}
\end{figure}

In the second experiment, we consider the following four correlated models: 1) $\vLambda \vLambda^T=\vI$; 2) $\vLambda \vLambda^T=\vT(1/4)$; 3) $\vLambda \vLambda^T=\vT(1/2)$; 4) $\vLambda \vLambda^T=\vT(3/4)$. Let $n$ increase from $1$ to $30$ and $\eta=0.2$. We also make 500 Monte Carlo trials and calculate the average sample size for each fixed $n$ like the first experiment.

Fig. \ref{fig: sim2} reports the simulation results. We can easily see that the number of samples in the four cases is a linear function of the signal dimension $n$, which agrees with theoretical results (\ref{leq: I}) and (\ref{leq: T}). In addition, the larger the parameter $\theta$, the bigger the slope, which demonstrates that for a given estimation accuracy of $\vSigma$, the correlated case requires more samples than that of the less correlated one.
\begin{figure}
	\centering
	\includegraphics[width=3.2in]{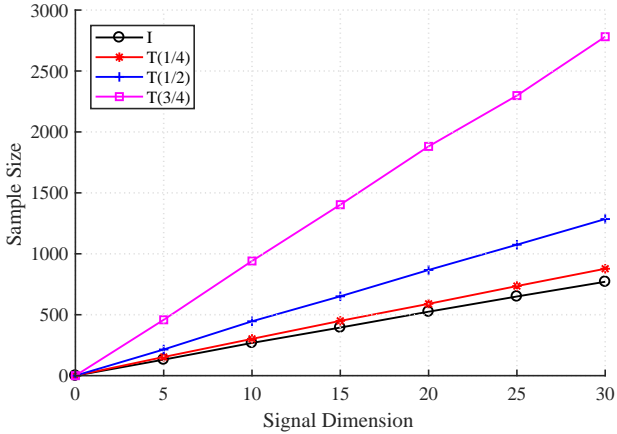}
	\caption{Sample size v.s. signal dimension for different correlated models.}
	\label{fig: sim2}
\end{figure}

In the third experiment, we compare the divergence rate of theoretical results (Theorem \ref{thm: General}) and Monte Carlo simulation results for four correlated models: 1) $\vLambda \vLambda^T=\vI$; 2) $\vLambda \vLambda^T=\vT(1/4)$; 3) $\vLambda \vLambda^T=\vT(1/2)$; 4) $\vLambda \vLambda^T=\vT(3/4)$. We set $n=15$ and increase $m$ from 50 to 1000 with step 50. For a fixed sample size $m$, we make 500 Monte Carlo trials and calculate the logarithm (base 10) of the average of estimation error $||\hat{\vSigma}-\vSigma||$.

The results are presented in Fig. \ref{fig: sim3}. From these results, we can know that for both theoretical and simulation results, the curves of four models are nearly parallel, which means that the four models have very similar error divergence rate. The results agree with Theorem \ref{thm: General}. However, we also observe a big gap between theoretical estimation errors and simulation estimation errors. This is because we have made a number of loose estimates in order to obtain a clear statement of the proofs, which leads to the fact that our theoretical bounds are not optimal in terms of multiplicative coefficients.

\begin{figure}
	\centering
	\includegraphics[width=3.2in]{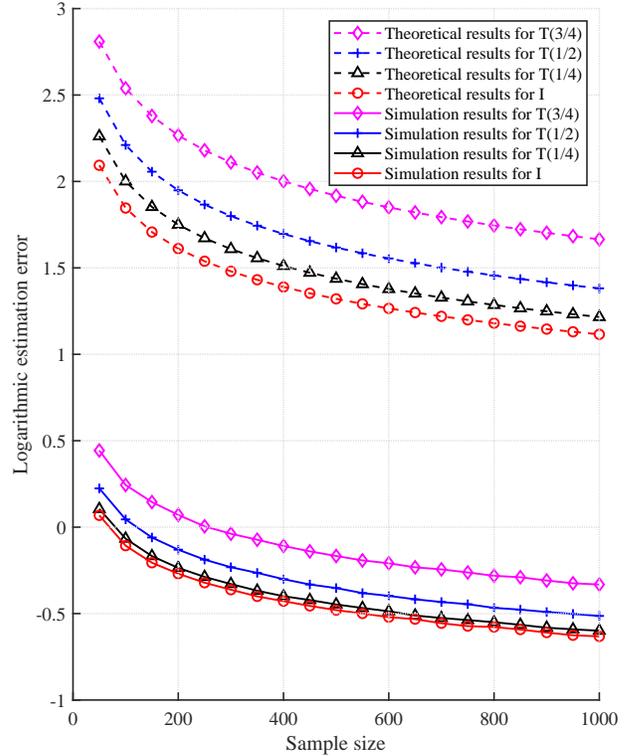}
	\caption{Logarithmic estimation error v.s. sample size for theoretical and simulation results under different models.}
	\label{fig: sim3}
\end{figure}

\section{conclusion} \label{sec: Conclusion}
In this paper, we have presented a non-asymptotic analysis for the covariance matrix estimation from { linearly-}correlated Gaussian samples. Our theoretical results have shown that the error bounds depend on the signal dimension $n$, the sample size $m$, and the shape parameter $\vB$ of the distribution of the correlated sample covariance matrix. In particular, when the shape parameter is a class of Toeplitz matrix (which is of great practical interest), $O(n)$ samples are sufficient to faithfully estimate the covariance matrix from correlated samples. This result has demonstrated that it is possible to estimate covariance matrices from moderate correlated samples.

For future work, it would be of great practical interest to extend the theoretical analysis for correlated samples from Gaussian distribution to other distributions such as sub-Gaussian, heavy tailed, and log-concave distributions. { In addition, it is of great importance to investigate the performance of other estimators under correlated samples. Examples include structured estimators, regularized estimators, and so on.}


%

\appendices

\section{Proof of Theorem \ref{lm: wishart1}} \label{sec: Proof1}
To prove Theorem \ref{lm: wishart1}, we require some useful definitions and facts. Without loss of generality, we assume $\vSigma= \vI$, otherwise we can use $\vSigma^{-1/2} \vX$ instead of ${\vX}$ to verify the general case.

\begin{definition}[$\e$-net]
	Let $\mathcal{K} \subset \R^n$ and $\e>0$. A subset $\mathcal{N} \subset \mathcal{K}$ is called an $\e$-net of $\mathcal{K}$ if
	$$\forall~\vx \in \mathcal{K}, ~~\exists~\vx_0 \in \mathcal{N} ~\text{such that}~ \norm{\vx-\vx_0}_2 \le \e.$$
\end{definition}

\begin{definition}[sub-Gaussian] A random variable $x$ is said to be sub-Gaussian with variance proxy $\sigma^2$ (denoted as $x \sim \text{\rm{subG}}(\sigma^2)$) if $\,\E x=0 $ and
	$$
	\E\exp(sx) \le \exp\(\frac{\sigma^2s^2}{2}\), ~ \forall \,\, s \in \R.
	$$	
\end{definition}
{ Equivalent definitions of sub-Gaussian random variables can be found in \cite[Proposition 2.5.2]{vershynin2017high}. Typical sub-Gaussian random variables include Gaussian variables, Bernoulli variables and any bounded random variables.
}
\begin{definition}[sub-exponential] A random variable $x$ is said to be sub-exponential with parameter $K$ (denoted as $x \sim \text{\rm{subE}}(K)$) if $\,\E x=0 $ and
	$$
	\E\exp(sx) \le \exp\(\frac{K^2s^2}{2}\), ~ \forall \,\, |s|\le \frac{1}{K}.
	$$
\end{definition}
{ Equivalent definitions of sub-exponential random variables can be found in \cite[Proposition 2.7.1]{vershynin2017high}. All sub-Gaussian random variables and their squares are sub-exponential. In addition, exponential and Poisson random variables belong to sub-exponential variables.
}

\begin{fact}[Exercise 4.4.3 and Corollary 4.2.13, \cite{vershynin2017high}] \label{lm: covering} Let $\vA \in \R^{m \times n}$ and $\e \in [0,1/2)$. Then for any $\e$-net $\mathcal{N}$ of the unit sphere $\mathbb{S}^{n-1}$ and $\e$-net $\mathcal{M}$ of the unit sphere $\mathbb{S}^{m-1}$, we have
	$$
	\norm{\vA} \le \frac{1}{1-2\e} \sup \limits_{\vx \in \mathcal{N}, \vy \in \mathcal{M}} \ip{\vA \vx}{\vy},
	$$
	{where $\ip{\cdot}{\cdot}$ denotes the inner product between two vectors, i.e. $\ip{\va}{\vb}=\va^T\vb.$}
	If $m=n$ and $\vA$ is symmetric, then
	$$
	\norm{\vA} \le \frac{1}{1-2\e} \sup \limits_{\vx \in \mathcal{N}} |\ip{\vA \vx}{\vx}|.
	$$
	Furthermore, there exist $\e$-nets $\mathcal{N}$ and $\mathcal{M}$ with cardinalities
	$$
	|\mathcal{N}| \le \(1+\frac{2}{\e}\)^n~~  \textrm{and} ~~|\mathcal{M}| \le \(1+\frac{2}{\e}\)^m.
	$$
\end{fact}

\begin{fact}[Lemma 1.4, \cite{rigollet2018high}] \label{lm: subE_property}
	Let $x\sim \text{\rm{subG}} (\sigma^2)$, then for any $k\ge 1$,
	$$
	\E |x|^k \le (2\sigma^2)^\frac{k}{2} k \Gamma\(\frac{k}{2}\),
	$$
	where $\Gamma(z) = \int_{0}^{\infty}x^{z-1}e^{-x}dx$.
\end{fact}

\begin{fact}[Lemma 1.12, \cite{rigollet2018high}] \label{lm: subG_subE}
	Let $x \sim \text{\rm{subG}} (\sigma^2) $, then the random variable $z=x^2-\E [x^2]$ is sub-exponential with $z \sim \text{\rm{subE}}(8\sigma^2)$.
\end{fact}

\begin{fact}[Bernstein's inequality, Theorem 2.8.2, \cite{vershynin2017high}] \label{lm: Bernstein} Let $x_1, \ldots, x_m$ be independent random variables with $\E x_i =0$ and $x_i \sim \text{\rm{subE}}(K)$, and $\va=(a_1,\ldots,a_m)\in \R^m$. Define $S_m=\frac{1}{m} \sum_{i=1}^{m}a_ix_i$. Then for any $t>0$, we have
	\begin{equation*}
		\P\(\left|S_m\right|\ge t\) \le 2 \exp\(-\frac{1}{2} \min \left\{\frac{m^2t^2}{K^2 \norm{\va}_2^2}, \frac{mt}{K\norm{\va}_\infty}\right\}\),
	\end{equation*}
	where $\norm{\va}_2^2=\sum_{i=1}^{m}a_i^2$ and $\norm{\va}_\infty=\max_i|a_i|$.
\end{fact}
Here facts \ref{lm: subG_subE} and \ref{lm: Bernstein} are derived by slightly modifying the original results in \cite{rigollet2018high} and \cite{vershynin2017high} respectively. For the convenience of the reader, we include the detailed proofs in Appendix \ref{appdix: Proof_facts}.

We are now in position to prove Theorem \ref{lm: wishart1}. For clarity, the proof is divided into several steps.
\begin{itemize}
	\item[1)] {\bf{Problem reduction}}. Let $\vB = \vU\vD_B\vU^T$ be the spectral decomposition of the symmetric matrix $\vB$, where $\vD_{\vB}=\diag \{\l_1,\l_2,\ldots,\l_m\}$ is a diagonal matrix whose entries are the eigenvalues of $\vB$, and $\vU$ is an orthonormal matrix. Then we have
	\begin{align} \label{eq: Problem reduction}
		&~~~~\P(\norm{\vW-\E\vW}  \ge t) \nonumber\\
		&= 	\P\(\frac{1}{m}\norm{\vX \vU \vD_B \vU^T\vX^T-\E[\vX\vU \vD_B\vU^T \vX^T]}\ge t\) \nonumber\\
		&= 	\P\(\frac{1}{m}\norm{\vX \vD_B \vX^T-\E[\vX \vD_B \vX^T]}\ge t\) \nonumber\\
		&= 	\P\(\norm{\frac{1}{m} \sum_{i=1}^{m}\l_i \(\vx_i \vx_i^T-\vI_n\)}\ge t\),
	\end{align}
	where the second equality holds because the Gaussian matrix $\vX$ is orthogonally invariant.
	\item[2)] {\bf{Approximation}}. Choose $\e=1/4$. By Fact \ref{lm: covering}, we get
	\begin{align*}
		&\norm{\frac{1}{m} \sum_{i=1}^{m}\l_i \(\vx_i \vx_i^T-\vI_n\)} \\
		&\quad \quad\le 2 \sup \limits_{\vu \in \mathcal{N}} \left|\frac{1}{m} \sum_{i=1}^{m}\l_i\ip{ \(\vx_i \vx_i^T-\vI_n\) \vu}{\vu} \right|\\
		&\quad \quad= 2 \sup \limits_{\vu \in \mathcal{N}} \left|\frac{1}{m} \sum_{i=1}^{m}\l_i (\ip{ \vx_i}{\vu}^2-1)\right|,
	\end{align*}
	where $\mathcal{N}$ is a ${1}/{4}$-net of $\mathbb{S}^{n-1}$ with $|\mathcal{N}| \le 9^n$. Thus we have
	\begin{multline*}
		\P\(\norm{\frac{1}{m} \sum_{i=1}^{m}\l_i \(\vx_i \vx_i^T-\vI_n\)}\ge t\)
		\\ \le \P\( \sup \limits_{\vu \in \mathcal{N}} \left|\frac{1}{m} \sum_{i=1}^{m}\l_i (\ip{ \vx_i}{\vu}^2-1)\right| \ge \frac{t}{2} \).
	\end{multline*}
	\item[3)] {\bf{Concentration}}. Fix $\vu  \in \mathcal{N}$, we are going to bound
	$$\P\( \left|\frac{1}{m} \sum_{i=1}^{m}\l_i (\ip{\vx_i}{\vu}^2-1)\right| \ge \frac{t}{2} \).
	$$
	By assumption, $\ip{\vx_i}{\vu}$ are independent Gaussian random variables with mean zero and variance $\E\ip{\vx_i}{\vu}^2=1$. Thus we have $\ip{\vx_i}{\vu} \sim \text{\rm{subG}}(1)$. From Fact \ref{lm: subG_subE}, we know $(\ip{\vx_i}{\vu}^2-1)$ are independent sub-exponential variables with mean zero and $(\ip{\vx_i}{\vu}^2-1) \sim \text{subE}(8)$. By using Bernstein's inequality (Fact \ref{lm: Bernstein}), we have
	\begin{multline*}
		\P\( \left|\frac{1}{m} \sum_{i=1}^{m}\l_i (\ip{\vx_i}{\vu}^2-1)\right| \ge \frac{t}{2} \) \\
		\le 2 \exp\(-\frac{1}{32} \min \left\{\frac{m^2t^2}{16 \sum_{i=1}^{m}\l_i^2}, \frac{mt}{\max_i|\l_i|}\right\}\).
	\end{multline*}
	Since $\norm{\vB}_F^2=\sum_{i=1}^{m}\l_i^2$ and $\norm{\vB}=\max_i|\l_i|$, we obtain
	\begin{multline*}
		\P\( \left|\frac{1}{m} \sum_{i=1}^{m}\l_i (\ip{\vx_i}{\vu}^2-1)\right| \ge \frac{t}{2} \) \\
		\le 2 \exp\(-\frac{1}{32} \min \left\{\frac{m^2t^2}{16 \norm{\vB}_F^2}, \frac{mt}{ \norm{\vB}}\right\}\).
	\end{multline*}
	\item[4)] {\bf{Tail bound}}. Taking union bound for all $\vu \in \mathcal{N}$ yields
	\begin{multline*}
		\P\( \sup \limits_{\vu \in \mathcal{N}} \left|\frac{1}{m} \sum_{i=1}^{m}\l_i (\ip{ \vx_i}{\vu}^2-1)\right| \ge t/2 \) \\
		\le 9^n \cdot 2 \exp\(-\frac{1}{32} \min \left\{\frac{m^2t^2}{16 \norm{\vB}_F^2}, \frac{mt}{\norm{\vB}}\right\}\).
	\end{multline*}
	Assigning
	$$t=\frac{32\norm{\vB}_F\delta+64\norm{\vB}\delta^2}{m}:=t_1,$$ we obtain
	\begin{multline*}
		\P\(\sup \limits_{\vu \in \mathcal{N}} \left|\frac{1}{m} \sum_{i=1}^{m}\l_i (\ip{ \vx_i}{\vu}^2-1)\right| \ge \frac{t_1}{2}\) \\
		\le 9^n \cdot 2 \exp\(-2 \delta^2 \) = 2 \exp\(- 2\delta^2 + 2n\log3\).
	\end{multline*}

   Therefore, we show that, for any $\delta \geq 0$,
	\begin{multline*}
		\P\(\norm{\vW-\E\vW}  \ge \frac{32\norm{\vB}_F\delta+64\norm{\vB}\delta^2}{m}\)\\
		\le 2 \exp\(- 2\delta^2 + 2n\log3\).
	\end{multline*}

   In particular, if $\delta\ge \sqrt{2\log3}\sqrt{n}$, we have
	\begin{multline*}
		\P\(\norm{\vW-\E\vW}  \ge \frac{32\norm{\vB}_F\delta+64\norm{\vB}\delta^2}{m}\)\\
		\le 2 \exp\(- \delta^2\),
	\end{multline*}
   which is useful to establish the expectation bound.

	\item[5)] {\bf{Expectation bound}}. For any $\l \ge \sqrt {2\log3} \sqrt{n}$, we have
	\begin{align*}
		&~~~~\E\norm{\vW-\E\vW }\\
		&= \int_0^\infty {\mathbb{P}(\norm{\vW-\E\vW} \ge t) {\rm{d}} t}\\
		&=\frac{32}{m} \int_0^\infty {\mathbb{P}(\norm{\vW-\E\vW} \ge t)\(\norm{\vB}_{F} + 4\delta\norm{\vB}\) {\rm{d}} \delta}\\
		&\le \frac{32}{m} \int_0^\l 1 \cdot (\norm{\vB}_{F}+ 4\delta \norm{\vB}) {\rm{d}} \delta \\
		&~~~ + \frac{64}{m}\int_\l^\infty \exp \(-\delta^2\) (\norm{\vB}_{F}+ 4\delta \norm{\vB}) {\rm{d}} \delta,
	\end{align*}
	where the first equality follows from the integral identity, in the second inequality we have let $t = (32\norm{\vB}_F\delta+64\norm{\vB}\delta^2)/{m}$.
	We continue by scaling and variable replacement as follows
	\begin{align*}
		&~~~~\E\norm{\vW-\E\vW }\\
		& \le \frac{32}{m}(\norm{\vB}_{F} \l + 2\norm{\vB} \l^2)\\
		&~~~ + \frac{64}{m}\int_\l^\infty \exp \(-\delta^2\) \(\frac{\delta\norm{\vB}_{F}}{\lambda}+ 4\delta \norm{\vB}\) {\rm{d}} \delta\\
		& \le \frac{32}{m}(\norm{\vB}_{F} \l + 2\norm{\vB} \l^2)\\
		&~~~ + \frac{32}{m}\int_0^\infty \exp \( - x\) \(\frac{\norm{\vB}_{F}}{\l}+ 4\norm{\vB}\) {\rm{d}} x\\
		& = \frac{32}{m}\(\norm{\vB}_{F} \l + 2\norm{\vB} \l^2 +\frac{\norm{\vB}_{F}}{\l}+ 4\norm{\vB}\)\\
		& \le \frac{32}{m}\[\norm{\vB}_{F} \l + 2\norm{\vB} \l^2 +\(\frac{\norm{\vB}_{F}}{\l}+ 4\norm{\vB}\)\cdot \frac{\l^2}{2}\]\\
		& \le \frac{48\norm{\vB}_{F} \l + 128\norm{\vB} \l^2}{m}.
	\end{align*}
	The last inequality follows from $ \lambda^2 \ge 2(\log 3) n \ge 2$. Choosing $\l = \sqrt {2\log3} \sqrt{n} $, we obtain
	\begin{align} \label{neq: E}
		&\E\norm{\vW-\E\vW } \\&\le \frac{48 \sqrt{2\log 3} \norm{\vB}_{F} \sqrt{n} + 256 \log 3 \norm{\vB} n}{m}\\
		&\le \frac{72\norm{\vB}_{F} \sqrt{n} + 282 \norm{\vB} n}{m},
	\end{align}
	which completes the proof.
\end{itemize}

\section{Proof of Facts} \label{appdix: Proof_facts}

\subsection{Proof of Fact \ref{lm: subG_subE}}
In this proof, we slightly improve the result of \cite[Lemma 1.12]{rigollet2018high}. By the Taylor expansion, we have
$$
\E[\exp(s(x^2-\E [x^2]))] = 1 + \sum_{k=2}^{\infty} \frac{s^k \E[x^2-\E [x^2]]^k}{k!}.
$$
Due to the convexity of $x^k$ for $x>0$ and $k \ge 1$, it follows from Jensen's inequality that
$$
\(\frac{x^2-\E [x^2]}{2}\)^k \le \(\frac{x^2+\E [x^2]}{2}\)^k \le \frac{x^{2k}+(\E [x^2])^k}{2}.
$$
By using the above inequality and Jensen's inequality again, we obtain
\begin{align*}
\E \exp(s(x^2-\E [x^2])) &\le 1 + \sum_{k=2}^{\infty} \frac{s^k 2^{k-1} \(\E[x^{2k}]+\(\E [x^2]\)^k\)}{k!}\\
&\le 1 + \sum_{k=2}^{\infty} \frac{s^k 2^{k} \E[x^{2k}]}{k!}.
\end{align*}
By using Fact \ref{lm: subE_property}, if $|s| \le \frac{1}{8\sigma^2}$, we have
\begin{align*}
\E \exp(s(x^2-\E [x^2])) &\le 1 + \sum_{k=2}^{\infty} \frac{s^k 2^{k} (2\sigma^2)^k k!}{k!}\\
&\le 1 + \sum_{k=2}^{\infty}(4s\sigma^2)^k \\
&\le 1 +32 s^2 \sigma^4\\
&\le \exp\(\frac{(8\sigma^2)^2s^2}{2}\).
\end{align*}
According to the definition of sub-exponential random variable, we have
$(x^2-\E [x^2]) \sim \text{\rm{subE}}(8 \sigma^2)$.

\subsection{Proof of Fact \ref{lm: Bernstein}}
The proof is developed from \cite[Theorem 2.8.2]{vershynin2017high} and \cite[Theorem 1.13]{rigollet2018high} with explicit constant.
Without loss of generality, we assume that $K=1$, otherwise we can replace $x_i$ by $x_i/K$ and $t$ by $t/K$ to verify the general result. By using the Chernoff bound, for all $s>0$, we have
\begin{align*}
	\P(S_m \ge t ) &\le \exp(-smt) \E \exp\(s \sum_{i=1}^{m}a_i x_i\)\\
	&= \exp(-smt) \prod_{i=1}^{m} \E \exp\(s a_i x_i\).
\end{align*}
According to the definition of sub-exponential, if $|s|\le {1}/{|a_i|}$, we have
$$
	\E \exp(s a_i x_i) \le \exp\(\frac{s^2a_i^2}{2}\).
$$
In order to make the above inequality hold for all $i$, we have $|s| \le 1/\norm{\va}_\infty$. So we have
\begin{align*}
\P(S_m \ge t ) &= \exp(-smt) \prod_{i=1}^{m} \exp\(\frac{s^2a_i^2}{2}\)\\
&=\exp\(\frac{\norm{\va}_2^2}{2}s^2-smt\).
\end{align*}	
Choosing $$s=\min\left\{\frac{mt}{\norm{\va}_2^2},\frac{1}{\norm{\va}_\infty}\right\}$$
yields
	\begin{equation*}
	\P\(S_m\ge t\) \le \exp\(-\frac{1}{2} \min \left\{\frac{m^2t^2}{K^2 \norm{\va}_2^2}, \frac{mt}{K\norm{\va}_\infty}\right\}\).
	\end{equation*}
We can obtain the same bound for $\P\(S_m \le -t\)$ by replacing $x_i$ by $-x_i$, which completes the proof.

\section*{Acknowledgment}
The authors thank the anonymous referees and the Associate Editor for useful comments which have helped to improve the presentation of this paper.

\ifCLASSOPTIONcaptionsoff
  \newpage
\fi


\begin{thebibliography}{10}
	\providecommand{\url}[1]{#1}
	\csname url@samestyle\endcsname
	\providecommand{\newblock}{\relax}
	\providecommand{\bibinfo}[2]{#2}
	\providecommand{\BIBentrySTDinterwordspacing}{\spaceskip=0pt\relax}
	\providecommand{\BIBentryALTinterwordstretchfactor}{4}
	\providecommand{\BIBentryALTinterwordspacing}{\spaceskip=\fontdimen2\font plus
		\BIBentryALTinterwordstretchfactor\fontdimen3\font minus
		\fontdimen4\font\relax}
	\providecommand{\BIBforeignlanguage}[2]{{%
			\expandafter\ifx\csname l@#1\endcsname\relax
			\typeout{** WARNING: IEEEtran.bst: No hyphenation pattern has been}%
			\typeout{** loaded for the language `#1'. Using the pattern for}%
			\typeout{** the default language instead.}%
			\else
			\language=\csname l@#1\endcsname
			\fi
			#2}}
	\providecommand{\BIBdecl}{\relax}
	\BIBdecl
	
	\bibitem{krim1996two}
	H.~Krim and M.~Viberg, ``Two decades of array signal processing research: the
	parametric approach,'' \emph{IEEE Signal Process. Mag.}, vol.~13, no.~4, pp.
	67--94, 1996.
	
	\bibitem{friedman2001elements}
	T.~Hastie, R.~Tibshirani, and J.~Friedman, \emph{The Elements of Statistical
		Learning, 2nd ed.}\hskip 1em plus 0.5em minus 0.4em\relax New York, NY, USA:
	Springer, 2009, ch. 4 and 14.
	
	\bibitem{cai2016estimating}
	T.~T. Cai, Z.~Ren, and H.~H. Zhou, ``Estimating structured high-dimensional
	covariance and precision matrices: Optimal rates and adaptive estimation,''
	\emph{Electron. J. Stat.}, vol.~10, no.~1, pp. 1--59, 2016.
	
	\bibitem{fan2016overview}
	J.~Fan, Y.~Liao, and H.~Liu, ``An overview of the estimation of large
	covariance and precision matrices,'' \emph{Econom. J.}, vol.~19, no.~1, pp.
	C1--C32, 2016.
	
	\bibitem{capon1969high}
	J.~Capon, ``High-resolution frequency-wavenumber spectrum analysis,''
	\emph{Proc. IEEE}, vol.~57, no.~8, pp. 1408--1418, 1969.
	
	\bibitem{schmidt1986multiple}
	R.~Schmidt, ``Multiple emitter location and signal parameter estimation,''
	\emph{{IEEE} Trans. Antennas Propag.}, vol.~34, no.~3, pp. 276--280, 1986.
	
	\bibitem{roy1989esprit}
	R.~Roy and T.~Kailath, ``{ESPRIT}-estimation of signal parameters via
	rotational invariance techniques,'' \emph{{IEEE} Trans. Acoust., Speech,
		Signal Process.}, vol.~37, no.~7, pp. 984--995, 1989.
	
	\bibitem{bai1993limit}
	Z.~Bai and Y.~Yin, ``Limit of the smallest eigenvalue of a large dimensional
	sample covariance matrix,'' \emph{Ann. Probab.}, pp. 1275--1294, 1993.
	
	\bibitem{aubrun2007sampling}
	G.~Aubrun, ``Sampling convex bodies: a random matrix approach,'' \emph{Proc.
		Amer. Math. Soc.}, vol. 135, no.~5, pp. 1293--1303, 2007.
	
	\bibitem{vershynin2010into}
	R.~Vershynin, ``Introduction to the non-asymptotic analysis of random
	matrices,'' in \emph{Compressed Sensing, Theory and Applications}, Y.~Eldar
	and G.~Kutyniok, Eds.\hskip 1em plus 0.5em minus 0.4em\relax Cambridge, U.K.:
	Cambridge Univ. Press., 2012, pp. 201--268.
	
	\bibitem{adamczak2010quantitative}
	R.~Adamczak, A.~Litvak, A.~Pajor, and N.~Tomczak-Jaegermann, ``Quantitative
	estimates of the convergence of the empirical covariance matrix in
	log-concave ensembles,'' \emph{J. Amer. Math. Soc.}, vol.~23, no.~2, pp.
	535--561, 2010.
	
	\bibitem{adamczak2011sharp}
	R.~Adamczak, A.~E. Litvak, A.~Pajor, and N.~Tomczak-Jaegermann, ``Sharp bounds
	on the rate of convergence of the empirical covariance matrix,'' \emph{C.R.
		Math.}, vol. 349, no.~3, pp. 195--200, 2011.
	
	\bibitem{vershynin2012close}
	R.~Vershynin, ``How close is the sample covariance matrix to the actual
	covariance matrix?'' \emph{J. Theor. Probab.}, vol.~25, no.~3, pp. 655--686,
	2012.
	
	\bibitem{srivastava2013covariance}
	N.~Srivastava and R.~Vershynin, ``Covariance estimation for distributions with
	2+$\varepsilon$ moments,'' \emph{Ann. Probab.}, vol.~41, no.~5, pp.
	3081--3111, 2013.
	
	\bibitem{koltchinskii2017concentration}
	V.~Koltchinskii and K.~Lounici, ``Concentration inequalities and moment bounds
	for sample covariance operators,'' \emph{Bernoulli}, vol.~23, no.~1, pp.
	110--133, 2017.
	
	\bibitem{ramirez2010detection}
	D.~Ram{\'\i}rez, J.~V{\'\i}a, I.~Santamar{\'\i}a, and L.~L. Scharf, ``Detection
	of spatially correlated {Gaussian} time series,'' \emph{IEEE Trans. Signal
		Process.}, vol.~58, no.~10, pp. 5006--5015, 2010.
	
	\bibitem{huang2013detection}
	Y.~Huang and X.~Huang, ``Detection of temporally correlated signals over
	multipath fading channels,'' \emph{IEEE Trans. Wireless Commun.}, vol.~12,
	no.~3, pp. 1290--1299, 2013.
	
	\bibitem{shiu2000fading}
	D.-S. Shiu, G.~J. Foschini, M.~J. Gans, and J.~M. Kahn, ``Fading correlation
	and its effect on the capacity of multielement antenna systems,'' \emph{IEEE
		Trans. Commun.}, vol.~48, no.~3, pp. 502--513, 2000.
	
	\bibitem{liu2007training}
	Y.~Liu, T.~F. Wong, and W.~W. Hager, ``{Training signal design for estimation
		of correlated MIMO channels with colored interference},'' \emph{IEEE Trans.
		Signal Process.}, vol.~55, no.~4, pp. 1486--1497, 2007.
	
	\bibitem{epps1979comovements}
	T.~W. Epps, ``Comovements in stock prices in the very short run,'' \emph{J.
		Amer. Stat. Assoc.}, vol.~74, no. 366a, pp. 291--298, 1979.
	
	\bibitem{munnix2010impact}
	M.~C. M{\"u}nnix, R.~Sch{\"a}fer, and T.~Guhr, ``Impact of the tick-size on
	financial returns and correlations,'' \emph{Physica A}, vol. 389, no.~21, pp.
	4828--4843, 2010.
	
	\bibitem{bickel2008regularized}
	P.~J. Bickel and E.~Levina, ``Regularized estimation of large covariance
	matrices,'' \emph{Ann. Stat.}, pp. 199--227, 2008.
	
	\bibitem{cai2013optimal}
	T.~T. Cai, Z.~Ren, and H.~H. Zhou, ``Optimal rates of convergence for
	estimating toeplitz covariance matrices,'' \emph{Probab. Theory Relat.
		Fields}, vol. 156, no. 1-2, pp. 101--143, 2013.
	
	\bibitem{bickel2008covariance}
	P.~J. Bickel and E.~Levina, ``Covariance regularization by thresholding,''
	\emph{Ann. Stat.}, pp. 2577--2604, 2008.
	
	\bibitem{rothman2009generalized}
	A.~J. Rothman, E.~Levina, and J.~Zhu, ``Generalized thresholding of large
	covariance matrices,'' \emph{J. Amer. Stat. Assoc.}, vol. 104, no. 485, pp.
	177--186, 2009.
	
	\bibitem{carlson1988covariance}
	B.~D. Carlson, ``Covariance matrix estimation errors and diagonal loading in
	adaptive arrays,'' \emph{IEEE Trans. Aerosp. Electron. Syst.}, vol.~24,
	no.~4, pp. 397--401, 1988.
	
	\bibitem{ledoit2004well}
	O.~Ledoit and M.~Wolf, ``A well-conditioned estimator for large-dimensional
	covariance matrices,'' \emph{J. of Multivariate Anal.}, vol.~88, no.~2, pp.
	365--411, 2004.
	
	\bibitem{chen2010shrinkage}
	Y.~Chen, A.~Wiesel, Y.~C. Eldar, and A.~O. Hero, ``Shrinkage algorithms for
	mmse covariance estimation,'' \emph{IEEE Trans. Signal Process.}, vol.~58,
	no.~10, pp. 5016--5029, 2010.
	
	\bibitem{coluccia2015regularized}
	A.~Coluccia, ``Regularized covariance matrix estimation via empirical bayes,''
	\emph{IEEE Signal Process. Lett.}, vol.~22, no.~11, pp. 2127--2131, 2015.
	
	\bibitem{stein1986lectures}
	C.~Stein, ``Lectures on the theory of estimation of many parameters,'' \emph{J.
		Sov. Math.}, vol.~34, no.~1, pp. 1373--1403, 1986.
	
	\bibitem{el2008spectrum}
	N.~El~Karoui, ``Spectrum estimation for large dimensional covariance matrices
	using random matrix theory,'' \emph{Ann. Stat.}, vol.~36, no.~6, pp.
	2757--2790, 2008.
	
	\bibitem{ledoit2011eigenvectors}
	O.~Ledoit and S.~P{\'e}ch{\'e}, ``Eigenvectors of some large sample covariance
	matrix ensembles,'' \emph{Probab. Theory Relat. Fields}, vol. 151, no. 1-2,
	pp. 233--264, 2011.
	
	\bibitem{collins2013compound}
	B.~Collins, D.~McDonald, and N.~Saad, ``Compound {Wishart} matrices and noisy
	covariance matrices: {Risk underestimation},'' 2013. [Online]. Available:
	https://arxiv.org/abs/1306.5510.
	
	\bibitem{burda2011applying}
	Z.~Burda, A.~Jarosz, M.~A. Nowak, J.~Jurkiewicz, G.~Papp, and I.~Zahed,
	``Applying free random variables to random matrix analysis of financial data.
	{Part I: The Gaussian case},'' \emph{Quant. Financ.}, vol.~11, no.~7, pp.
	1103--1124, 2011.
	
	\bibitem{chuah2002capacity}
	C.-N. Chuah, D.~N.~C. Tse, J.~M. Kahn, and R.~A. Valenzuela, ``Capacity scaling
	in {MIMO} wireless systems under correlated fading,'' \emph{IEEE Trans. Inf.
		Theory}, vol.~48, no.~3, pp. 637--650, 2002.
	
	\bibitem{speicher1998combinatorial}
	R.~Speicher, \emph{Combinatorial Theory of the Free Product with Amalgamation
		and Operator-valued Free Probability Theory}.\hskip 1em plus 0.5em minus
	0.4em\relax Rhode Island, USA: American Mathematical Society, 1998.
	
	\bibitem{soloveychik2014error}
	I.~Soloveychik, ``Error bound for compound wishart matrices,'' 2014. [Online].
	Available: https://arxiv.org/abs/1402.5581.
	
	\bibitem{levina2012partial}
	E.~Levina and R.~Vershynin, ``Partial estimation of covariance matrices,''
	\emph{Probab. Theory Relat. Fields}, vol. 153, no. 3-4, pp. 405--419, 2012.
	
	\bibitem{paulin2016efron}
	D.~Paulin, L.~Mackey, and J.~A. Tropp, ``{Efron-Stein} inequalities for random
	matrices,'' \emph{Ann. Probab.}, vol.~44, no.~5, pp. 3431--3473, 2016.
	
	\bibitem{stein1972bound}
	C.~Stein, ``A bound for the error in the normal approximation to the
	distribution of a sum of dependent random variables,'' in \emph{Proc.
		Berkeley Symposium Mathematical Statistics and Probability}.\hskip 1em plus
	0.5em minus 0.4em\relax The Regents of the University of California, 1972,
	pp. 583--602.
	
	\bibitem{stein1986approximate}
	------, ``Approximate computation of expectations,'' \emph{Lecture
		Notes-Monograph Series}, vol.~7, pp. 1--164, 1986.
	
	\bibitem{bryc2008compound}
	W.~Bryc, ``Compound real {Wishart and q-Wishart} matrices,'' \emph{Int. Math.
		Res. Notices}, vol. 2008, 2008.
	
	\bibitem{golub2012matrix}
	G.~H. Golub and C.~F. Van~Loan, \emph{Matrix Computations, 4th ed.}\hskip 1em
	plus 0.5em minus 0.4em\relax Maryland, USA: Johns Hopkins Univ. Press, 2013.
	
	\bibitem{vershynin2017high}
	R.~Vershynin, \emph{High-Dimensional Probability An Introduction with
		Applications in Data Science}.\hskip 1em plus 0.5em minus 0.4em\relax
	Cambridge, U.K.: Cambridge Univ. Press, 2018.
	
	\bibitem{rigollet2018high}
	P.~Rigollet, ``High-dimensional statistics,'' \emph{Lecture notes for course
		18.S997}, 2018.
	
\end{thebibliography}

\end{document}